\documentclass[11pt]{article}

\usepackage[T1]{fontenc}
\usepackage{lmodern}
\usepackage{microtype}
\usepackage[a4paper,margin=28mm]{geometry}
\usepackage{amsmath,amssymb,amsthm,mathtools}
\usepackage{enumitem}
\usepackage{authblk}
\usepackage{xcolor}
\usepackage{hyperref}
\usepackage[nameinlink,capitalize,noabbrev]{cleveref}
\usepackage{tikz}
 \usetikzlibrary{
	   arrows.meta,
	  decorations.pathreplacing,
	  shapes.geometric
 }
 
\definecolor{linkblue}{RGB}{28,72,120}
\hypersetup{
  colorlinks=true,
  linkcolor=linkblue,
  citecolor=linkblue,
  urlcolor=linkblue,
  pdftitle={Independent Set Discovery on Biclique-Free Graphs Is Fixed-Parameter Tractable}
}
\newtheorem{theorem}{Theorem}[section]
\newtheorem{lemma}[theorem]{Lemma}
\newtheorem{corollary}[theorem]{Corollary}
\theoremstyle{definition}
\newtheorem{openproblem}[theorem]{Open Problem}
\newtheorem{definition}[theorem]{Definition}
\crefname{lemma}{Lemma}{Lemmas}
\Crefname{lemma}{Lemma}{Lemmas}
\crefname{corollary}{Corollary}{Corollaries}
\Crefname{corollary}{Corollary}{Corollaries}

\newcommand{\dist}{\operatorname{dist}}
\newcommand{\OPT}{\operatorname{OPT}}
\newcommand{\ISD}{\textnormal{\textsc{Independent Set Discovery}}}
\newcommand{\WIT}{\textnormal{\textsc{Weighted Independent Transversal}}}

\title{\textbf{Independent Set Discovery on Biclique-Free Graphs\\
Is Fixed-Parameter Tractable}}
\author[1]{Chenghua Liu}
\author[2]{Boning Meng}

\affil[1]{Institute of Software, Chinese Academy of Sciences, Beijing, China}
\affil[2]{University of Regensburg, Regensburg, Germany}
\affil[ ]{\texttt{liuch.russell@gmail.com}, \texttt{mengboning2013@gmail.com}}

\date{}

\begin{document}
\maketitle

\begin{abstract}
\ISD{} asks whether a configuration of $k$ tokens on distinct vertices can be
transformed into an independent $k$-set by a sequence of token slides, each
moving one token to an unoccupied neighbor; only the terminal configuration
must be independent.  \ISD{} is a central problem in
solution discovery: its target is not prescribed and must be chosen together
with the token movements needed to reach it.
Fellows et al.~\cite{FellowsEtAl2026} proved it FPT in $k$ on every fixed
bounded-degeneracy class and every nowhere-dense class, leaving the
biclique-free case open.  The biclique-free setting lies substantially beyond
both regimes:
biclique-free classes can have unbounded degeneracy and even be somewhere
dense.  We resolve the open problem affirmatively.  Given
an $n$-vertex, $m$-edge graph promised to be $K_{d,d}$-free and an initial
$k$-token configuration, our deterministic algorithm computes the minimum
number of slides and, in time $2^{O(dk\log k)}(n+m)^{O(1)}$, returns an optimal
independent target and a shortest collision-free slide sequence or certifies
that no independent target is reachable.  Thus the problem is FPT in $k$ for
every fixed $d$ and uniformly FPT in $k+d$.  The proof combines an exact
minimum-cost assignment characterization of token movement with local
branching on bounded, cost-relevant \emph{cheap prefixes} of candidate lists.
The method also yields an exact FPT algorithm for weighted independent
transversals on $K_{d,d}$-free graphs with overlapping candidate sets, direct
exact FPT algorithms for both problems on bounded-degeneracy graphs, an
edge-count-sensitive XP algorithm, sharper bounds for bounded $s$-codegree and
unbalanced biclique exclusion, and an exact extension to weighted movement on
a separate directed graph.
\end{abstract}
\section*{Acknowledgments}
The authors used OpenAI's ChatGPT in preparing
this manuscript, including for language editing and \LaTeX{} preparation.
ChatGPT also contributed to the exploratory development of most of the
arguments. All claims and proofs were independently checked and finalized
by the authors, who take full responsibility for the content.
\section{Introduction}

Solution discovery asks for a feasible solution that is close to a prescribed,
possibly infeasible, state.  Fellows et al.~\cite{FellowsEtAl2026} formalized
this viewpoint through reconfiguration moves.  In \ISD{}, the state is a set
$S$ of $k$ occupied vertices of a graph $G$.  A move slides one token along an
edge to an unoccupied vertex, and the objective is to reach an independent
$k$-set using at most a given number of moves.  Only the terminal configuration
must be independent.  This endpoint-only feasibility requirement distinguishes
discovery from independent-set reconfiguration and is essential to the
algorithm developed here.

The framework captures a natural repair problem: the target is not prescribed,
and the cost is the actual movement needed to reach a feasible state.  \ISD{}
is one of the four fundamental problems used by Fellows et al.\ to introduce
solution discovery~\cite{FellowsEtAl2026}.  It is a basic test case for the
framework: pairwise nonadjacency is the canonical packing constraint, yet the
target independent set must be chosen jointly with the movement of the tokens.
Consequently, endpoint feasibility is coupled to global shortest-path
information from the initial configuration.

Since its introduction, solution discovery has developed along several
complementary directions.  Grobler et al.\ extended the framework to spanning
trees, shortest paths, matchings, and cuts, and classified these problems
under several token-movement models~\cite{GroblerEtAlICALP2024}.  Subsequent
work initiated the systematic study of kernelization, including polynomial
kernels for \ISD{} parameterized by $k$ on nowhere-dense
classes~\cite{GroblerEtAl2024}.  Saito et al.\ obtained clique-width-based
algorithms for \ISD{}, hardness results on chordal and diameter-two graphs,
and polynomial-time solvability on split graphs~\cite{SaitoEtAl2026}.  The
framework has also led to logical meta-theorems, cost--value objectives, and
models separating the feasibility graph from the movement
graph~\cite{BousquetEtAl2025,GerhardEtAl2026,vonBergenEtAl2026}.  This breadth,
and the repeated appearance of \ISD{} across these developments, make its
structural parameterized complexity a basic question for the emerging theory.

\subsection{The biclique-free frontier}

Fellows et al.~\cite{FellowsEtAl2026} proved that \ISD{} is NP-complete on
planar graphs of maximum degree four.  Parameterized by the combined quantity
$k+b$, it remains W[1]-hard even when $C_4,\ldots,C_p$ are excluded as induced
subgraphs for any fixed $p$; parameterized by the slide budget $b$ alone, it
is W[1]-hard on 2-degenerate bipartite graphs.  On the positive side, they
gave an XP algorithm parameterized by treewidth and proved FPT in $k$ on every
fixed bounded-degeneracy class and on every nowhere-dense class.  Subsequent
work broadened the surrounding landscape through kernelization, width-based
upper and lower bounds, and logical
meta-theorems~\cite{GroblerEtAl2024,SaitoEtAl2026,BousquetEtAl2025}.

Biclique exclusion is a natural and substantially broader frontier.  Every
$a$-degenerate graph is $K_{a+1,a+1}$-free, but the converse fails sharply.
For every prime $q$, the incidence graph of the projective plane
$PG(2,q)$ is $C_4$-free, $(q+1)$-regular, and hence has degeneracy $q+1$,
despite having $2(q^2+q+1)$ vertices~\cite{KuhnLokshtanovMiller2024}.
Moreover, the $1$-subdivision of
$K_t$ is $C_4$-free for every $t$, and $K_t$ is its $1$-shallow minor; hence
the class of $C_4$-free graphs is somewhere dense.
Thus neither earlier FPT theorem covers the full biclique-free regime.  At
the same time, excluding a fixed biclique imposes strong extremal
structure~\cite{KovariSosTuran1954}.  A positive result would therefore show
that this extremal sparsity alone suffices for fixed-parameter
tractability, well beyond the previously understood sparse classes.

There is also a concrete obstacle to extending the earlier algorithm.  It
constructs a $k$-independence-covering family---a family $\mathcal F$ of
independent sets such that every independent set of size at most $k$ is
contained in some member of $\mathcal F$---and solves one minimum-cost
matching instance for each
member~\cite{FellowsEtAl2026,LokshtanovEtAl2020}.  Kuhn, Lokshtanov, and Miller
proved that, for every function $f$ and every fixed $\varepsilon>0$, no bound
of the form $f(k)n^{k/4-\varepsilon}$ holds for such families on all
$C_4$-free graphs~\cite{KuhnLokshtanovMiller2024}.  Since $C_4=K_{2,2}$,
explicitly enumerating these global families cannot yield an FPT algorithm
even in the first nontrivial biclique-free case.  An FPT resolution therefore
cannot proceed by explicitly enumerating such a global covering family and
needs a different mechanism for representing candidate targets.

Against this background, Fellows et al.\ explicitly posed the following
question~\cite{FellowsEtAl2026}.

\begin{openproblem}[Fellows et al.~\cite{FellowsEtAl2026}]
\label{op:biclique-free-isd}
For every fixed integer $d\ge2$, is \ISD{} fixed-parameter tractable
parameterized by the number $k$ of tokens on $K_{d,d}$-free graphs?
\end{openproblem}

\subsection{Our results}
For an initial configuration $S$, let $\OPT(G,S)$ be the minimum number of
slides needed to reach an independent set of size $|S|$, and put
$\OPT(G,S)=+\infty$ if no such set is reachable.  Our main result resolves
\cref{op:biclique-free-isd} and strengthens the requested result in two ways:
the dependence is uniform in $k+d$, and the algorithm computes the optimum
together with an optimal slide sequence.

\begin{theorem}[Main theorem]
\label{thm:main}
Let $d,k\ge2$ be integers.  Given an $n$-vertex, $m$-edge graph $G$ promised
to be $K_{d,d}$-free and a set $S\in\binom{V(G)}{k}$, one can deterministically
compute $\OPT(G,S)$.  If $\OPT(G,S)<+\infty$, the algorithm returns an
independent $k$-set $T$ and a collision-free sequence of exactly
$\OPT(G,S)$ token slides from $S$ to $T$; otherwise, it reports that no
independent $k$-set is reachable from $S$.  Its running time is
\[
  2^{O(dk\log k)}(n+m)^{O(1)}.
\]
Consequently, the budgeted decision version is FPT in $k$ for every fixed
$d$ and uniformly FPT in $k+d$.  It is obtained by comparing $\OPT(G,S)$
with the binary-encoded budget $b$, which adds only polynomial dependence on
the encoding length of $b$.
\end{theorem}

The proof yields three further contributions.
\begin{enumerate}[leftmargin=2em,itemsep=0.35em]
\item \emph{A general weighted static theorem.}
We solve a weighted independent-transversal problem in which the $k$
candidate sets may overlap and the cost of a vertex depends on its label.
\Cref{thm:weighted} gives the exact optimization algorithm, and
\cref{cor:algorithmic-transfer} transfers both the optimum and a witness back
to \ISD{}.

\item \emph{A direct proof for bounded-degeneracy graphs.}
The same search gives a short direct proof of FPT on bounded-degeneracy
classes.  In contrast to the covering-family proof, it uses only a local edge
count and prefix branching.  It also gives an explicit uniform running time
in the number of tokens and the degeneracy, works for the weighted problem,
and yields exact optimization and sequence reconstruction.  Thus the method
recovers the earlier tractability theorem while strengthening its algorithmic
formulation; see \cref{thm:degenerate}.

\item \emph{Further algorithms and sharper bounds.}
The local conflict bound yields an edge-count-sensitive XP algorithm on
arbitrary graphs and sharper uniform FPT bounds for bounded $s$-codegree (the
maximum common-neighborhood size of an $s$-set) and unbalanced $K_{s,t}$
exclusion.
The latter bound depends exponentially on the smaller side $s$, whereas
applying the balanced theorem with $d=t$ would put the larger side in the
extremal exponent.  Finally, the assignment argument extends to a two-graph
model with a biclique-free feasibility graph and an
arbitrary directed, nonnegatively weighted movement graph.
These consequences are stated in
\cref{thm:edge-number,thm:codegree,thm:unbalanced-biclique,cor:two-graph}.
\end{enumerate}

\subsection{Proof overview}

The proof separates movement from terminal feasibility.  Temporarily label
the tokens by their starting positions.  For any fixed target $T$, the minimum
number of slides needed to reach $T$ equals the cost of a minimum assignment
of starting positions to vertices of $T$, where assignment costs are
shortest-path distances.  A simple exchange argument handles occupied
vertices along the assigned paths: whenever a token blocks another token's
path, the two tokens exchange destinations in the analysis.  The assignment
cost does not increase, and every actual slide still enters an unoccupied
vertex.  Thus the dynamic problem reduces exactly to selecting pairwise
nonadjacent destinations of minimum total assignment cost.

We formulate that static task as a weighted independent-transversal problem.
Each starting token is a coordinate, its reachable vertices are possible
destinations, and two selected destinations conflict if they are equal or
adjacent.  The destination lists may overlap and costs depend on the
coordinate, so a direct application of a standard transversal theorem is not
enough.

The algorithm uses a \emph{cheap-prefix dichotomy}.  Suppose that $r$ labels
remain.  For each label, form a prefix of its $(dr)^{O(d)}$ cheapest
destinations compatible with the partial assignment.  If some residual
candidate list is shorter than this prefix length, we branch on its entire list.
Otherwise all prefixes are full.  If the sum of their maximum costs fits the
residual budget, the K\H{o}v\'ari--S\'os--Tur\'an bound and the
independent-transversal criterion
of Wanless and Wood~\cite{KovariSosTuran1954,WanlessWood2022} guarantee
compatible representatives.  If this sum exceeds the budget, every feasible
solution must select a vertex from at least one prefix, because a solution
avoiding all prefixes would have cost at least the sum of their coordinate-wise
lower bounds.  Branching over the prefixes is therefore exhaustive.  Each branch
fixes one label, so the depth is at most $k$ and the branching factor depends
only on $k+d$.

The decisive simplification is locality.  Extremal sparsity is invoked only
on bounded, cost-relevant prefixes, rather than through a global family that
must cover every small independent set.  This is precisely why the method
continues to work in the $C_4$-free setting where small global covering
families need not exist.

Viewed statically, our problem is a graph-metric instance of \emph{Dispersion}
in the movement-minimization
framework~\cite{DemaineEtAl2009,DemaineHajiaghayiMarx2014,BiloEtAl2016}.  Our
contribution in this setting is an exact FPT algorithm under biclique-free
terminal conflicts, allowing overlapping destination lists and
label-dependent costs and reconstructing an optimal token-slide sequence.

\paragraph{Organization.}
\Cref{sec:preliminaries} defines the discovery model and the weighted static
problem.  \Cref{sec:assignment} proves the assignment-distance equivalence.
\Cref{sec:algorithm} then gives the biclique-free cheap-prefix algorithm and
completes the proof of \cref{thm:main}.  Finally,
\cref{sec:sparsity-extensions} develops the additional sparse-graph and
weighted-movement consequences.

\section{Preliminaries}
\label{sec:preliminaries}

We first formalize the discovery model and then define the weighted static
problem used by the algorithm.

\subsection{Configurations and Independent Set Discovery}

\begin{definition}[Configuration and token slide]
\label{def:configuration-slide}
Let $G$ be a finite simple undirected graph and let $k$ be a nonnegative
integer.  A \emph{$k$-token configuration} is a set $Q\subseteq V(G)$ with
$|Q|=k$.  Its vertices are \emph{occupied}.  For two $k$-token configurations
$Q$ and $Q'$, write $Q\rightarrow_G Q'$ if there are $u\in Q$ and
$v\in V(G)\setminus Q$ such that $uv\in E(G)$ and
\[
  Q'=(Q\setminus\{u\})\cup\{v\}.
\]
The transition $Q\rightarrow_G Q'$ is a \emph{token slide}.  A
\emph{collision-free slide sequence} of length $\ell$ is a sequence
$Q_0,Q_1,\ldots,Q_\ell$ of $k$-token configurations such that
$Q_{j-1}\rightarrow_G Q_j$ for every $j\in\{1,\ldots,\ell\}$.
\end{definition}

Configurations are sets, so a slide must enter an unoccupied vertex.  The
tokens are unlabeled; labels will be introduced only as an algorithmic device
for assigning starting positions to destinations.

\begin{definition}[Independent Set Discovery]
\label{def:isd}
An instance of \ISD{} is a triple $(G,S,b)$, where $G$ is a finite simple
undirected graph, $S\subseteq V(G)$ is an initial configuration, and $b$ is a
nonnegative integer encoded in binary.  Let $k=|S|$.  The instance is a
yes-instance if there exist an integer $\ell$ satisfying $0\le\ell\le b$, an
independent $k$-token configuration $T$, and a collision-free slide sequence
\[
  S=Q_0,Q_1,\ldots,Q_\ell=T.
\]
Only the terminal configuration $Q_\ell=T$ is required to be independent.
The optimization version omits $b$ and, given $(G,S)$, asks for the minimum
possible $\ell$; its value is $\OPT(G,S)$, with $\OPT(G,S)=+\infty$ if no
independent $k$-set is reachable from $S$.
\end{definition}

\subsection{Numbers, graphs, and distances}

We write $\mathbb N=\{1,2,\ldots\}$, $\mathbb N_0=\mathbb N\cup\{0\}$, and
$[q]=\{1,\ldots,q\}$ for $q\in\mathbb N$.  All graphs are finite, simple,
and undirected unless a directed graph is explicitly mentioned.  For a graph
$G$, its vertex and edge sets are $V(G)$ and $E(G)$; we usually write
$n=|V(G)|$ and $m=|E(G)|$.  A set $I\subseteq V(G)$ is \emph{independent} if
no edge of $G$ has both endpoints in $I$.

For $X\subseteq V(G)$, let
\[
N_G(X)=\{v\in V(G)\setminus X: uv\in E(G)\text{ for some }u\in X\},
\qquad N_G[X]=X\cup N_G(X).
\]
For a single vertex, we write $N_G(v)$ and $N_G[v]$.  The subgraph induced by
$X$ is denoted by $G[X]$.  For disjoint $X,Y\subseteq V(G)$, let $e_G(X,Y)$
be the number of edges with one endpoint in $X$ and the other in $Y$.  Graph
subscripts are omitted when the underlying graph is clear.

A path has length equal to its number of edges.  For $u,v\in V(G)$, the
distance $\dist_G(u,v)$ is the length of a shortest $u$--$v$ path, and is
defined to be $+\infty$ when $u$ and $v$ lie in different connected
components.  In an unweighted graph, all distances from
one source, together with corresponding shortest-path predecessors, can be
computed by breadth-first search in $O(n+m)$ time.

For $d\in\mathbb N$, $K_{d,d}$ is the complete bipartite graph with $d$
vertices in each part.  A graph is $K_{d,d}$-free when it has no subgraph
isomorphic to $K_{d,d}$.  This is stronger than forbidding an induced
$K_{d,d}$.  Edges within either side of the biclique are immaterial because
the forbidden copy need not be induced.

For our biclique-free algorithms, $d$ is part of the input and
$K_{d,d}$-freeness is a promise: the algorithms need not verify it.

\subsection{Independent transversals and overlapping candidate sets}

Let $C_1,\ldots,C_r$ be pairwise disjoint vertex sets in a graph $H$.  An
\emph{independent transversal} is an independent set containing exactly one
vertex from each $C_i$.  We will use a variant in which each label $i$ has a
\emph{candidate set} $A_i$ of admissible vertices, and these candidate sets
are not required to be disjoint.  The intended interpretation is that label
$i$ represents a token and $A_i$ contains its possible destinations; the same
vertex may therefore belong to the candidate sets of several tokens.  A valid
selection chooses one vertex $x_i\in A_i$ for each label so that the selected
vertices are distinct and pairwise nonadjacent.  The following definition
adds coordinate-dependent costs and a budget to this version with overlapping
candidate sets.

\begin{definition}[Weighted independent transversal]
\label{def:wit}
For $k\in\mathbb N$, an instance
$\mathcal I=(G,(A_i,c_i)_{i\in[k]},B)$ of the decision version of \WIT{}
consists of
\begin{itemize}[leftmargin=2em,itemsep=0.15em]
\item a graph $G$;
\item an explicitly listed finite candidate set $A_i\subseteq V(G)$ for every
label $i\in[k]$, where candidate sets may overlap;
\item a nonnegative rational cost $c_i(v)$ for every $v\in A_i$; and
\item a nonnegative rational budget $B$.
\end{itemize}
A \emph{feasible transversal} is a tuple $(x_1,\ldots,x_k)$ with
$x_i\in A_i$ for every $i$, all entries distinct, and no two entries adjacent
in $G$.  Its cost is $\sum_{i=1}^k c_i(x_i)$.  The decision problem asks
whether a feasible transversal of cost at most $B$ exists.  An optimization
instance omits $B$ and asks for a minimum-cost feasible transversal, or for a
report that none exists.  We write $|\mathcal I|$ for the total binary
encoding length in either version.  We use
\emph{label} and \emph{coordinate} interchangeably for an index $i\in[k]$.
\end{definition}

All integers and rational costs in a \WIT{} instance are encoded in binary;
the same holds for the budget $B$ when it is present.  Arithmetic operations
and comparisons have polynomial bit complexity.

In the reduction from \ISD{}, $A_i$ will contain the vertices reachable from
the initial position $s_i$, with $c_i(v)=\dist_G(s_i,v)$.

\section{From token slides to weighted transversals}
\label{sec:assignment}

To reduce \ISD{} to \WIT{}, we temporarily label the tokens and assign one
destination to each label.  We now prove that the resulting static assignment
cost is exactly the number of slides needed by the original unlabeled tokens.
The key freedom is that, if one token blocks another token's assigned shortest
path, the two tokens may exchange their assigned destinations.

All costs in this section are interpreted in the extended nonnegative
integers.  We use the conventions
\[
a+(+\infty)=+\infty
\qquad\text{and}\qquad
\min\varnothing=+\infty.
\]

For two sets $Q,T\subseteq V(G)$ of equal cardinality, define their
\emph{assignment distance}
\[
\mu_G(Q,T)=
\min_{\pi:Q\to T\text{ bijective}}
\sum_{u\in Q}\dist_G(u,\pi(u)).
\]
Equivalently, $\mu_G(Q,T)$ is the cost of a minimum-weight perfect matching
between the starting vertices $Q$ and target vertices $T$, where the edge
$ut$ has weight $\dist_G(u,t)$.  An infinite weight means that the two
vertices lie in different connected components.
The following matching characterization is known in closely related
unlabeled-pebble models~\cite{CalinescuDumitrescuPach2008,FellowsEtAl2026}.
The proof is constructive.  Its central idea is to follow an assigned
shortest path only to its first unoccupied vertex.  If another token blocks
that path, the two tokens exchange their assigned destinations in the
analysis; this exchange is bookkeeping and does not contribute to the number of slides.

\begin{lemma}[Assignment-distance lemma]
\label{lem:assignment}
Let $G$ be a finite simple undirected graph, and let $Q,T\subseteq V(G)$ have
equal cardinality.  The minimum length of a collision-free token-slide
sequence from $Q$ to $T$ equals $\mu_G(Q,T)$, where this minimum is
$+\infty$ if no such sequence exists; intermediate configurations are
unrestricted beyond collision-freedom.
\end{lemma}

\begin{proof}
Temporarily label every token by its initial vertex.  Any slide sequence
induces a bijection $\pi:Q\to T$.  A token starting at $u$ traverses at least
$\dist_G(u,\pi(u))$ edges, so every sequence has length at least
$\mu_G(Q,T)$.  In particular, no sequence exists when
$\mu_G(Q,T)=+\infty$.

Suppose that $\mu_G(Q,T)<+\infty$.  We prove by induction on $\mu_G(Q,T)$
that a sequence of exactly $\mu_G(Q,T)$ slides exists.  The case
$\mu_G(Q,T)=0$ is immediate.  Otherwise $Q\ne T$.  Choose an unoccupied
target $t\in T\setminus Q$, fix an optimal bijection $\pi:Q\to T$, and let
$s=\pi^{-1}(t)$ be the token assigned to $t$.  Along a shortest
$s$--$t$ path, let $v$ be the first unoccupied vertex and let $u$ be its
predecessor.  Such a vertex exists because $t\notin Q$.  By the choice of
$v$, the predecessor $u$ is occupied and $uv\in E(G)$.

If $u\ne s$, let $t_u=\pi(u)$ and exchange the destinations assigned to the
tokens at $s$ and $u$.  Because $u$ lies on a shortest $s$--$t$ path, the
triangle inequality gives
\begin{align*}
 \dist_G(u,t)+\dist_G(s,t_u)
 &\le \dist_G(u,t)+\dist_G(s,u)+\dist_G(u,t_u)\\
 &=\dist_G(s,t)+\dist_G(u,t_u).
\end{align*}
Thus the exchanged assignment costs no more than the optimal assignment
$\pi$, and hence it is also optimal.  If $u=s$, no exchange is needed.  In
either case, we now have an optimal assignment that sends the token at $u$
to $t$.

Slide that token from $u$ to the empty neighbor $v$, and write
$Q'=(Q\setminus\{u\})\cup\{v\}$.  The suffix of the chosen path from $v$ to
$t$ is shortest, so replacing source $u$ by source $v$ in the current
assignment gives
\[
 \mu_G(Q',T)\le \mu_G(Q,T)-1.
\]
Conversely, take an optimal assignment from $Q'$ to $T$ and replace its
source $v$ by the adjacent source $u$, keeping the same destination.  The
triangle inequality gives
\[
 \mu_G(Q,T)\le 1+\mu_G(Q',T).
\]
Therefore $\mu_G(Q',T)=\mu_G(Q,T)-1$.  The induction hypothesis transforms
$Q'$ into $T$ in that many further slides.  Prepending the slide $u\to v$
produces an optimal sequence from $Q$ to $T$.  Every actual slide enters an
unoccupied vertex, so the constructed sequence is collision-free.
\end{proof}

\begin{corollary}[Static endpoint formulation]
\label{cor:static-formulation}
Let $(G,S,b)$ be an instance of \ISD{}, where
$S=\{s_1,\ldots,s_k\}$, and define
$c_i:V(G)\to\mathbb N_0\cup\{+\infty\}$ by
$c_i(v)=\dist_G(s_i,v)$.  The instance is feasible if and only if there are
distinct, pairwise nonadjacent vertices $x_1,\ldots,x_k\in V(G)$ such that
\[
 \sum_{i=1}^k c_i(x_i)\le b.
\]
Moreover,
\begin{equation}
 \OPT(G,S)=
 \min_{\substack{x_1,\ldots,x_k\in V(G)\text{ distinct}\\
                  \{x_i,x_j\}\notin E(G)\ (i\ne j)}}
 \sum_{i=1}^k c_i(x_i).
 \label{eq:global-opt-static}
\end{equation}
\end{corollary}

\begin{proof}
Labeling tokens by their starting vertices, every slide sequence ending at an
independent target yields a feasible tuple whose assignment cost is at most
the sequence length.  Conversely, a feasible tuple defines the independent
set $T=\{x_1,\ldots,x_k\}$, and the assignment $s_i\mapsto x_i$ gives
$\mu_G(S,T)\le\sum_i c_i(x_i)$.  \Cref{lem:assignment} then yields a slide
sequence of length at most this sum.  Taking minima proves both assertions.
\end{proof}

\begin{corollary}[Algorithmic transfer from \WIT{} to \ISD{}]
\label{cor:algorithmic-transfer}
Let $\mathcal C_p$ be a promise class of graphs indexed by a parameter $p$.
Suppose that, given any \WIT{} optimization instance $\mathcal I$ with $k$
candidate sets and underlying graph in $\mathcal C_p$, a minimum-cost
transversal or a report of nonexistence can be computed in time
\[
 f(k,p)|\mathcal I|^{O(1)}.
\]
Then, given a graph $G\in\mathcal C_p$ and a set
$S\in\binom{V(G)}{k}$, one can compute $\OPT(G,S)$ in time
\[
 f(k,p)(n+m)^{O(1)}.
\]
If the optimum is finite, the algorithm also returns an optimal independent
target and a shortest collision-free slide sequence; otherwise, it reports
nonexistence.
Here $n=|V(G)|$ and $m=|E(G)|$.
The statement remains valid when $p$ is an input-dependent quantity such as
$m=|E(G)|$.
\end{corollary}

\begin{proof}
Write $S=\{s_1,\ldots,s_k\}$.  For label $i$, take as its candidate set the
vertices in the connected component of $s_i$, and give candidate $v$ cost
$c_i(v)=\dist_G(s_i,v)$.  These distances and the candidate sets are obtained
by $k$ breadth-first searches.  By \cref{cor:static-formulation}, the optimum
of this \WIT{} instance is exactly $\OPT(G,S)$.

If the \WIT{} algorithm reports that no transversal exists, then
\cref{cor:static-formulation} gives $\OPT(G,S)=+\infty$, and we stop.
Otherwise, let $(x_1,\ldots,x_k)$ be an optimal tuple and
$T=\{x_1,\ldots,x_k\}$.  Its pairing $s_i\mapsto x_i$ is a minimum-cost
matching between $S$ and $T$: a cheaper bijection to the same target would be
a cheaper feasible tuple.  The constructive argument in
\cref{lem:assignment} can therefore be implemented to output exactly
$\OPT(G,S)$ slides.  For an explicit polynomial
implementation, precompute all-pairs unweighted distances and one next edge
on every finite shortest path.  In each iteration, scan the current assigned
shortest path until its first unoccupied vertex, perform the destination
exchange from the proof if needed, and make one slide.  Every assigned path
is simple, so the scan uses at most $n-1$ edges; moreover any finite optimum
is at most $k(n-1)$.  Thus distance computation and sequence reconstruction
add only a polynomial factor.
\end{proof}

\section{The cheap-prefix algorithm}
\label{sec:algorithm}

We first derive a biclique-free transversal certificate and then use it in the
cheap-prefix recursion for \WIT{}.

\subsection{Sparse conflicts and the threshold}
\label{sec:sparse-transversal}

Before defining the recursion, we prove the structural statement that makes
its bounded branching possible.  Suppose that each of $r$ labels has a
candidate set of $t$ vertices.  We want to select one vertex from every
candidate set so that the selected vertices are distinct and pairwise
nonadjacent.  The candidate sets may overlap, so both equality and adjacency
count as conflicts.

We next translate the K\H{o}v\'ari--S\'os--Tur\'an
bound~\cite{KovariSosTuran1954} into a bound on conflicts between two
possibly overlapping candidate sets.  Fix an integer $d\ge2$, a
$K_{d,d}$-free graph $G$, and an integer $t\ge1$.  Let
$X,Y\subseteq V(G)$ satisfy $|X|=|Y|=t$, and form disjoint labeled copies
\[
X^{\mathrm L}=\{x^{\mathrm L}:x\in X\},
\qquad
Y^{\mathrm R}=\{y^{\mathrm R}:y\in Y\}.
\]
Choosing one vertex from each candidate set produces one
\emph{cross-pair} in $X^{\mathrm L}\times Y^{\mathrm R}$.  Hence there are
exactly
\[
|X^{\mathrm L}\times Y^{\mathrm R}|=t^2
\]
possible cross-pairs.  Pairs with both endpoints in the same copied set are
irrelevant, since the two choices come from different candidate sets.

Let $J_{X,Y}$ be the bipartite graph with parts $X^{\mathrm L}$ and
$Y^{\mathrm R}$ in which
\[
x^{\mathrm L}y^{\mathrm R}\in E(J_{X,Y})
\quad\Longleftrightarrow\quad
xy\in E(G).
\]
Thus the edges of $J_{X,Y}$ are precisely the adjacency-conflict
cross-pairs.  The graph $J_{X,Y}$ is $K_{d,d}$-free.  Indeed, a
$K_{d,d}$ in $J_{X,Y}$ would project to two sets of $d$ distinct vertices
of $G$.  These two sets must be disjoint, since a vertex occurring on both
sides would require a loop in $G$.  Their projections would therefore form
a $K_{d,d}$ in $G$.

The balanced K\H{o}v\'ari--S\'os--Tur\'an bound consequently gives
\begin{equation}
	|E(J_{X,Y})|
	\le (d-1)^{1/d}t^{2-1/d}+(d-1)t.
	\label{eq:kst}
\end{equation}

There is one additional equality-conflict cross-pair
$x^{\mathrm L}x^{\mathrm R}$ for each $x\in X\cap Y$.  Its endpoints are
distinct copied vertices lying in different bipartition classes, so this
does not violate bipartiteness.  Since $|X\cap Y|\le t$, the total number of
adjacency-or-equality conflict cross-pairs is at most
\[
(d-1)^{1/d}t^{2-1/d}+(d-1)t+t.
\]
Accordingly, define
\[
p_d(t)=(d-1)^{1/d}t^{-1/d}+\frac{d}{t}.
\]
The proportion of the $t^2$ cross-pairs that are conflicts is then at most
\[
\frac{(d-1)^{1/d}t^{2-1/d}+(d-1)t+t}{t^2}
=p_d(t).
\]

We also use the following block-average crossing-degree theorem.  If
$C_1,\ldots,C_r$ is a partition of a graph, an edge is \emph{crossing at
$C_i$} when it has exactly one endpoint in $C_i$.

\begin{theorem}[Wanless--Wood criterion~\cite{WanlessWood2022}]
\label{thm:wanless-wood}
Let $H$ be a finite graph, let $r,t\in\mathbb N$, and let
$(C_1,\ldots,C_r)$ be a partition of $V(H)$ with $|C_i|\ge t$ for every
$i\in[r]$.  If
\[
 e_H(C_i,V(H)\setminus C_i)\le \frac t4|C_i|
 \qquad\text{for every }i\in[r],
\]
then $H$ has an independent transversal of $(C_1,\ldots,C_r)$.
\end{theorem}

\begin{lemma}[Biclique-free transversal lemma]
\label{lem:sparse-transversal}
Let $d,r,t\in\mathbb N$ with $d,r\ge2$, let $G$ be $K_{d,d}$-free, and let
$B_1,\ldots,B_r\subseteq V(G)$ be possibly overlapping sets of the same size
$t$.  If
\begin{equation}
 4(r-1)p_d(t)<1,
 \label{eq:transversal-certificate}
\end{equation}
then there are representatives $y_i\in B_i$ that are distinct and pairwise
nonadjacent.
\end{lemma}

\begin{proof}
For every $i$, make a disjoint copy
$C_i=\{(i,v):v\in B_i\}$.  Define an $r$-partite conflict graph $H$ on
$C_1\mathbin{\dot\cup}\cdots\mathbin{\dot\cup}C_r$ by joining $(i,u)$ and
$(j,v)$, for $i\ne j$, exactly when $u=v$ or $uv\in E(G)$.  Thus an
independent transversal of the copied blocks is exactly a tuple of distinct,
pairwise nonadjacent representatives of the original candidate sets.  Notice
that the copy $(i,v)$ records both the label $i$ and the original vertex $v$; the
copies are new vertices of $H$, not additional vertices of $G$.

Fix $i<j$.  By the projection argument accompanying~\eqref{eq:kst}, the
subgraph between $C_i$ and $C_j$ containing only conflicts
of the form $uv\in E(G)$ is $K_{d,d}$-free.  By~\eqref{eq:kst}, there
are at most
$(d-1)^{1/d}t^{2-1/d}+(d-1)t$ such edges.  Equality contributes at most
$|B_i\cap B_j|\le t$ further edges, so
\begin{equation}
 e_H(C_i,C_j)\le p_d(t)t^2.
 \label{eq:pair-conflict-density}
\end{equation}

For each $i$, the number of edges crossing at $C_i$ is
\[
 \sum_{j\ne i}e_H(C_i,C_j)
 \le (r-1)p_d(t)t^2
 <\frac{t^2}{4}=\frac t4|C_i|.
\]
The last strict inequality is exactly~\eqref{eq:transversal-certificate}.
The Wanless--Wood criterion therefore
gives an independent transversal of
$C_1,\ldots,C_r$, which projects to the required representatives.
\end{proof}

\begin{lemma}[Explicit prefix threshold]
\label{lem:explicit-threshold}
For integers $d,r\ge2$, define
\begin{equation}
 M(d,r)=(d-1)\bigl(4(r-1)\bigr)^d+4d^2(r-1).
 \label{eq:explicit-threshold}
\end{equation}
Then
\[
 4(r-1)p_d(M(d,r))<1.
\]
\end{lemma}

\begin{proof}
Put $Q=4(r-1)$ and set
\[
 t=M(d,r)=(d-1)Q^d+d^2Q,
 \qquad A=(d-1)^{1/d}Q,
 \qquad x=t^{1/d}.
\]
Since $t=x^d$ and $A=(d-1)^{1/d}Q$, we have
\[
 4(r-1)p_d(t)=Qp_d(t)=\frac{A}{x}+\frac{dQ}{x^d}.
\]
Thus it remains to prove that the last expression is less than one.  We have
$x>A$ and
$x^d-A^d=d^2Q$.  Factoring the difference of powers yields
\[
 d^2Q=(x-A)\sum_{j=0}^{d-1}x^{d-1-j}A^j
 <d x^{d-1}(x-A).
\]
Consequently $x^{d-1}(x-A)>dQ$, and division by $x^d=t$ gives
\[
 \frac Ax+\frac{dQ}{x^d}<1.
\]
The left-hand side is $Qp_d(t)$, proving the certificate.
\end{proof}

For integers $d,r\ge2$ and a $K_{d,d}$-free graph, the two lemmas imply that
arbitrary $M(d,r)$-element subsets of the $r$ candidate sets admit compatible
representatives.  For fixed $d$, $M(d,r)=\Theta_d(r^d)$.

\subsection{The cheap-prefix certificates}

A naive algorithm could choose a destination for each of the $k$ labels and
test the resulting $k$-tuple.  This takes roughly $n^k$ time.  The goal is to
replace the factor $n$ by a bound depending only on $d$ and the number $r$ of
labels that remain.

Think of a recursive state as a partial solution.  Its \emph{residual budget}
is the original budget minus the costs already paid, and the \emph{residual
candidate set} of an unassigned label consists of the vertices that do not
conflict with an already chosen destination and whose individual cost does
not exceed the residual budget.  The formal state below records exactly these
two objects.

Fix an arbitrary total order $\prec$ on $V(G)$, used only to break ties
between vertices of equal cost.  Order each residual candidate set by
nondecreasing coordinate cost, breaking ties according to $\prec$.  The first
\[
\min\{M(d,r),|A_i|\}
\]
vertices of the residual candidate set $A_i$ form the label's
\emph{cheap prefix}.
A cheap prefix is \emph{full} if it contains exactly $M(d,r)$ vertices.
The tie-breaking rule makes each prefix uniquely defined but has no effect on
the correctness of the algorithm.  Exactly one of the following situations
occurs.
\begin{enumerate}[label=\textup{(\roman*)},leftmargin=2.3em,itemsep=0.25em]
\item \emph{A residual candidate set is small.}  It has fewer than $M(d,r)$
vertices, so we branch over all vertices in that candidate set and thereby
assign the label.
\item \emph{All prefixes are full and collectively affordable.}  If the sum
of their largest costs is at most the residual budget, then every selection
from the prefixes is affordable.  By
\cref{lem:sparse-transversal,lem:explicit-threshold}, at least one such
selection is distinct and pairwise nonadjacent.  The decision procedure may
therefore return yes without searching below this state.
\item \emph{All prefixes are full but collectively too expensive.}  A
feasible solution cannot avoid every prefix: outside a prefix, the cost for
that coordinate is at least the prefix maximum, and these lower bounds already
exceed the budget.  Hence, in every feasible solution, at least one label is
assigned a vertex from its cheap prefix.  We branch on every possible such
label--vertex pair.
\end{enumerate}

Thus Case~(ii) certifies existence, whereas Case~(iii) certifies an exhaustive
branching rule; neither case makes a coordinate-wise greedy choice.

\begin{theorem}[Weighted transversal theorem]
\label{thm:weighted}
Let $d,k\ge2$ be integers, and let $\mathcal I$ be a \WIT{} instance with
$k$ candidate sets whose underlying graph is promised to be $K_{d,d}$-free.
The decision version can be solved deterministically, with a feasible tuple
returned for every yes-instance, in time
\begin{equation}
 O\!\left((kM(d,k))^k|\mathcal I|^{O(1)}\right)
 =2^{O(dk\log k)}|\mathcal I|^{O(1)}.
 \label{eq:weighted-runtime}
\end{equation}
The optimization version can be solved within the same parameterized bound:
the algorithm returns a globally minimum-cost transversal or reports that no
transversal exists.
\end{theorem}

The remainder of this section defines the recursion and proves this theorem.

\subsection{The recursive state}

For the decision version, fix a \WIT{} instance $\mathcal I$ with original
candidate sets $A_1^0,\ldots,A_k^0$, costs $c_i$, and original budget $B_0$.
A recursive state stores
\begin{itemize}[leftmargin=2em,itemsep=0.15em]
\item a set $D_F\subseteq[k]$ of already assigned labels;
\item a partial assignment $F:D_F\to V(G)$;
\item the unassigned labels $L=[k]\setminus D_F$;
\item the residual budget $B$; and
\item a residual candidate set $A_i$ for every $i\in L$.
\end{itemize}
Write $X_F=F(D_F)=\{F(i):i\in D_F\}$ for the set of vertices already chosen.
The state maintains
\begin{align}
 F(i)&\in A_i^0\ (i\in D_F),\quad
 F\text{ is injective, and }X_F\text{ is independent},
 \tag{I1}\label{inv:i1}\\
 B&=B_0-\sum_{i\in D_F}c_i(F(i)),
 \tag{I2}\label{inv:i2}\\
 A_j&=\{u\in A_j^0\setminus N[X_F]:c_j(u)\le B\}
 \quad(j\in L).
 \tag{I3}\label{inv:i3}
\end{align}
Invariant~\eqref{inv:i3} has two purposes.  Deleting $N[X_F]$ prevents a
future equality or adjacency conflict with a chosen vertex.  Deleting a
vertex of individual cost greater than $B$ is safe because all costs are
nonnegative.  Membership in $A_j$ certifies only that a single choice is
affordable, not that arbitrary choices from all residual candidate sets fit
the budget together.  At the root,
$F=\varnothing$, $L=[k]$, $B=B_0$, and
$A_i=\{v\in A_i^0:c_i(v)\le B_0\}$.

\subsection{The cheap-prefix dichotomy}

The next lemma states the branching argument formally.  Its affordable-prefix
part uses the transversal certificate already established in
\cref{lem:sparse-transversal,lem:explicit-threshold}; its over-budget part is
a sorted-cost lower bound.

At the current state, a \emph{feasible completion} is a tuple
$(x_i)_{i\in L}$ with $x_i\in A_i$ for every $i\in L$, with distinct and
pairwise nonadjacent entries, and with total residual cost at most $B$.

\begin{lemma}[Cheap-prefix dichotomy]
\label{lem:cheap-prefix}
Let $d\ge2$ be an integer.  Consider a state with $r=|L|\ge2$ in a
$K_{d,d}$-free graph,
and suppose $|A_i|\ge M(d,r)$ for every $i\in L$.  Let $P_i$ be the first
$M(d,r)$ vertices of $A_i$ in the fixed order extending nondecreasing cost,
and put
\[
 a_i=\max_{v\in P_i}c_i(v).
\]
Then the following statements hold.
\begin{enumerate}[label=\textup{(\alph*)},leftmargin=2.2em,itemsep=0.2em]
\item If $\sum_{i\in L}a_i\le B$, there are distinct, pairwise nonadjacent
vertices $y_i\in P_i$ for $i\in L$ such that
$\sum_{i\in L}c_i(y_i)\le B$.
\item If $\sum_{i\in L}a_i>B$, every feasible completion
$(x_i)_{i\in L}$ has $x_i\in P_i$ for at least one $i\in L$.
\end{enumerate}
\end{lemma}

\begin{proof}
For (a), \cref{lem:explicit-threshold} gives
$4(r-1)p_d(M(d,r))<1$.  Applying \cref{lem:sparse-transversal} to the sets
$(P_i)_{i\in L}$ supplies distinct, pairwise nonadjacent representatives
$y_i\in P_i$.  Each satisfies $c_i(y_i)\le a_i$, so their total cost is at
most $B$.

For (b), suppose a feasible completion avoids every prefix.  The ordering is
nondecreasing in cost, so $c_i(x_i)\ge a_i$ for every $i\in L$, including
when costs tie.  Therefore
$\sum_{i\in L}c_i(x_i)\ge\sum_{i\in L}a_i>B$, a contradiction.
\end{proof}

When a branch fixes $F(i)=v$, it spends $c_i(v)$, removes label $i$, and
updates every remaining candidate set to
\begin{equation}
 A_j'=\{u\in A_j\setminus N[v]:c_j(u)\le B'\},
 \qquad B'=B-c_i(v).
 \label{eq:candidate-set-update}
\end{equation}
Closed-neighborhood deletion simultaneously enforces distinctness and
nonadjacency.  Since the branch chooses $v\in A_i\subseteq A_i^0$, the new
assignment is also legal for label $i$.  Thus the update
preserves~\eqref{inv:i1}--\eqref{inv:i3}.

\subsection{The decision procedure}

The following box is the complete recursive solver.  ``Branch on $(i,v)$''
means create the child in~\eqref{eq:candidate-set-update}; return yes if some
child returns yes, and return no if every child returns no.

\begin{center}
\fbox{\begin{minipage}{0.93\linewidth}
\textbf{Cheap-prefix decision procedure at a state $(F,L,B,(A_i)_{i\in L})$.}
\begin{enumerate}[label=\textup{\arabic*.},leftmargin=2em,itemsep=0.25em,
                  topsep=0.35em]
\item Return no if $B<0$ or some $A_i$ is empty.  Return yes if
$L=\varnothing$ or $L=\{i\}$.
\item Put $r=|L|$.  For each $i\in L$, let $P_i$ be the first
$\min\{M(d,r),|A_i|\}$ vertices of $A_i$ in nondecreasing cost order.
\item If some $|A_i|<M(d,r)$, choose one such $i$ and branch on every
$(i,v)$ with $v\in A_i=P_i$.
\item Otherwise set $a_i=\max_{v\in P_i}c_i(v)$ for every $i\in L$.
If $\sum_{i\in L}a_i\le B$, return yes.
\item If $\sum_{i\in L}a_i>B$, branch on every $(i,v)$ with $i\in L$ and
$v\in P_i$.
\end{enumerate}
\end{minipage}}
\end{center}

Step~1 is valid for a single remaining label because its nonempty residual
candidate set already contains only vertices compatible with $X_F$ and of
cost at most $B$.  Step~3 enumerates an entire small candidate set.
Steps~4 and~5 are the two conclusions of \cref{lem:cheap-prefix}.  Thus every
terminal answer and every branch is justified by one of the explicit
certificates above.

Although Step~4 is nonconstructive, a standard self-reduction recovers a
witness: fix one label at a time by retaining the first tentative assignment
whose child remains feasible.  Thus no constructive implementation of the
Wanless--Wood theorem is needed.

\begin{lemma}[Correctness of the recursion]
\label{lem:recursion-correct}
Let $d\ge2$, and consider the cheap-prefix decision procedure above on a
$K_{d,d}$-free graph.  Every recursive state
satisfying~\eqref{inv:i1}--\eqref{inv:i3} returns yes if and only if there
exists a tuple
$(x_i)_{i\in L}\in\prod_{i\in L}A_i$ whose entries are distinct and pairwise
nonadjacent and whose total cost is at most $B$.
\end{lemma}

\begin{proof}
We induct on $r=|L|$.  If $B<0$, no tuple of nonnegative total cost is
feasible; if some residual candidate set is empty, no completion exists.  For
$r=0$, the empty residual tuple is feasible.  For $r=1$, every vertex in the
nonempty residual candidate set is compatible with $X_F$ and affordable
by~\eqref{inv:i3}.  These are exactly the base cases in Step~1.

For soundness, a Step~4 yes-answer is valid by
\cref{lem:cheap-prefix}(a).  In a branch, the induction hypothesis supplies a
child completion avoiding $N[v]$; adding $v$ yields a valid parent completion
of cost at most $B'+c_i(v)=B$.

For completeness, fix a feasible completion $(x_i)_{i\in L}$.  At a
state with a small candidate set, Step~3 explicitly contains the branch
$v=x_i$ for the chosen small label.  If all prefixes are full and affordable,
Step~4 returns yes.  If they are full and over budget,
\cref{lem:cheap-prefix}(b) gives a
label $i$ with $x_i\in P_i$, and Step~5 includes the branch $(i,x_i)$.
Independence places every other $x_j$ outside $N[x_i]$, and nonnegative costs
keep the remaining tuple within budget $B-c_i(x_i)$.  Hence the rest of the
fixed completion survives in the child, and induction applies.
\end{proof}

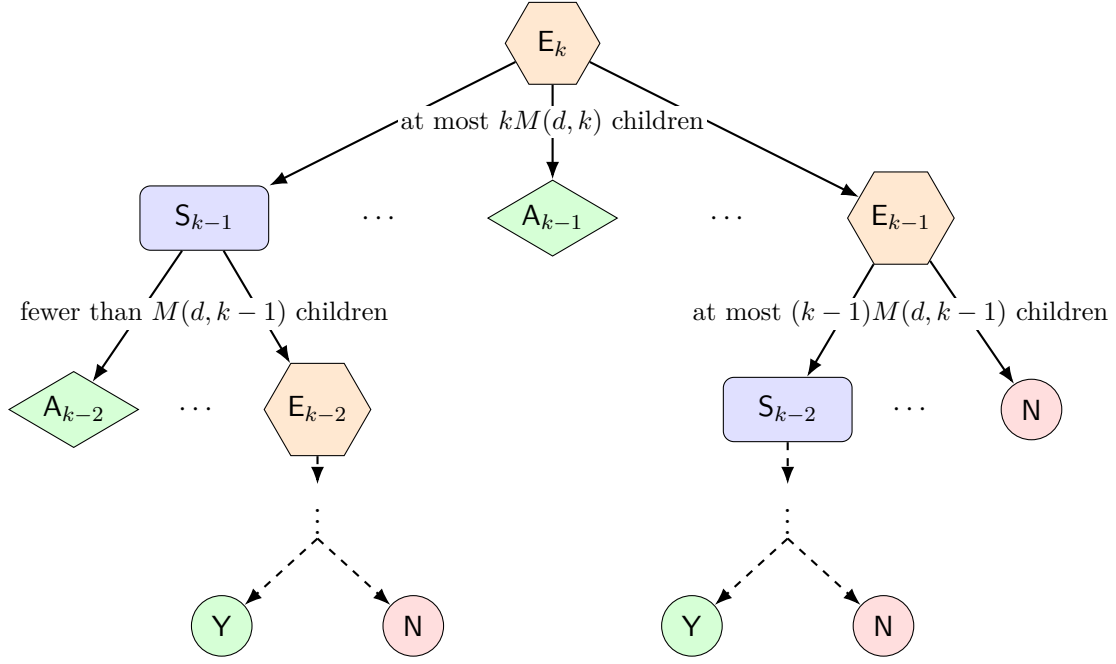
\begin{figure}[t]
	\centering
	\begin{tikzpicture}[
		x=1.15cm,
		y=1.1cm,
		>=Latex,
		edge/.style={->,thick},
		small/.style={
			draw,
			rounded corners,
			fill=blue!12,
			minimum width=1.7cm,
			minimum height=0.85cm
		},
		affordable/.style={
			draw,
			diamond,
			aspect=1.7,
			fill=green!16,
			inner sep=1.5pt,
			minimum width=1.5cm
		},
		expensive/.style={
			draw,
			regular polygon,
			regular polygon sides=6,
			fill=orange!18,
			minimum size=1.25cm,
			inner sep=1pt
		},
		noleaf/.style={
			draw,
			circle,
			fill=red!13,
			minimum size=0.8cm
		},
		yesleaf/.style={
			draw,
			circle,
			fill=green!16,
			minimum size=0.8cm
		},
		leveltext/.style={
			font=\small,
			fill=white,
			inner sep=1pt
		}
		]
		
		\node[expensive] (root) at (0,0) {\(\mathsf E_k\)};
		
		\node[small]      (s1) at (-4,-2.1) {\(\mathsf S_{k-1}\)};
		\node[affordable] (a1) at (0,-2.1)  {\(\mathsf A_{k-1}\)};
		\node[expensive]  (e1) at (4,-2.1)  {\(\mathsf E_{k-1}\)};
		
		\node at (-2,-2.1) {\(\cdots\)};
		\node at ( 2,-2.1) {\(\cdots\)};
		
		\draw[edge] (root) -- (s1);
		\draw[edge] (root) -- (a1);
		\draw[edge] (root) -- (e1);
		
		\node[leveltext] at (0,-0.95) {
			at most \(kM(d,k)\) children
		};
		
		\node[affordable] (a2) at (-5.5,-4.4) {\(\mathsf A_{k-2}\)};
		\node[expensive]  (e2) at (-2.7,-4.4) {\(\mathsf E_{k-2}\)};
		\node at (-4.1,-4.4) {\(\cdots\)};
		
		\draw[edge] (s1) -- (a2);
		\draw[edge] (s1) -- (e2);
		
		\node[leveltext] at (-4,-3.25) {
			fewer than \(M(d,k-1)\) children
		};
		
		\node[small]  (s2) at (2.7,-4.4) {\(\mathsf S_{k-2}\)};
		\node[noleaf] (n2) at (5.5,-4.4) {\(\mathsf N\)};
		\node at (4.1,-4.4) {\(\cdots\)};
		
		\draw[edge] (e1) -- (s2);
		\draw[edge] (e1) -- (n2);
		
		\node[leveltext] at (4,-3.25) {
			at most \((k-1)M(d,k-1)\) children
		};
		
		\node at (-2.7,-5.65) {\(\vdots\)};
		\node at ( 2.7,-5.65) {\(\vdots\)};
		
		\draw[edge,dashed] (e2) -- (-2.7,-5.3);
		\draw[edge,dashed] (s2) -- ( 2.7,-5.3);
		
		\node[yesleaf] (y3) at (-3.8,-7.0) {\(\mathsf Y\)};
		\node[noleaf]  (n3) at (-1.6,-7.0) {\(\mathsf N\)};
		\node[yesleaf]    (y4) at (1.6,-7.0) {\(\mathsf Y\)};
		\node[noleaf]     (n4) at (3.8,-7.0) {\(\mathsf N\)};
		
		\draw[edge,dashed] (-2.7,-5.95) -- (y3);
		\draw[edge,dashed] (-2.7,-5.95) -- (n3);
		\draw[edge,dashed] ( 2.7,-5.95) -- (y4);
		\draw[edge,dashed] ( 2.7,-5.95) -- (n4);
	\end{tikzpicture}
	
	\caption{A schematic bounded search tree for the \WIT{} algorithm.  A node
		\(\mathsf S_r\) represents case~\textup{(i)}: some residual candidate set
		\(A_i\) has size smaller than \(M(d,r)\), and the algorithm branches on all
		vertices of \(A_i\), producing fewer than \(M(d,r)\) children.  A node
		\(\mathsf A_r\) represents case~\textup{(ii)}: all cheap prefixes are full
		and collectively affordable.  By
		\cref{lem:sparse-transversal,lem:explicit-threshold}, this is a terminal
		yes-node.  A node \(\mathsf E_r\) represents case~\textup{(iii)}: all
		prefixes are full but collectively too expensive.  By
		\cref{lem:cheap-prefix}(b), every feasible completion meets some prefix, so
		the algorithm branches on at most \(rM(d,r)\) label--vertex pairs.  The
		symbols \(\mathsf Y\) and \(\mathsf N\) denote terminal yes- and no-states
		arising from the base cases.  Every recursive edge fixes one label, so \(r\)
		decreases by one.  Consequently, the tree has depth at most \(k\), maximum
		branching factor \(kM(d,k)\), and
		\(O((kM(d,k))^k)\) nodes.  The root is drawn as an
		\(\mathsf E_k\)-node only to illustrate the largest branching case; in an
		actual instance it may instead be an \(\mathsf S_k\)- or
		\(\mathsf A_k\)-node, or terminate immediately in a base case.}
	\label{fig:wit-bounded-search-tree}
\end{figure}
\subsection{Running time, witnesses, and optimization}

\begin{proof}[Proof of \Cref{thm:weighted}]
The root satisfies the invariants, so \cref{lem:recursion-correct} proves
budget feasibility.  Every branch assigns one label, and therefore the depth
is at most $k$.  The function $M(d,r)$ is nondecreasing in $r$.  A branching
state consequently has at most $rM(d,r)\le kM(d,k)$ children.  The search
tree has $O((kM(d,k))^k)$ states.
The three recursive state types and these depth and branching bounds are
summarized in \cref{fig:wit-bounded-search-tree}.

At each state, filtering candidate sets, selecting prefixes, testing equality
and adjacency, and performing exact rational arithmetic take polynomial time in
$|\mathcal I|$.  This proves the first bound
in~\eqref{eq:weighted-runtime}.  By~\eqref{eq:explicit-threshold},
\[
 M(d,k)\le(d-1)(4k)^d+4d^2k,
\]
and therefore
\[
 \log\bigl(kM(d,k)\bigr)=O(d\log k).
\]
Raising the branching factor to depth $k$ gives the second bound
in~\eqref{eq:weighted-runtime}.  A depth-first implementation stores only one
root-to-leaf assignment and its current candidate sets and prefixes.

To recover a witness, start at any state certified feasible by the decision
procedure.  Choose an unassigned label $i$ and test its residual vertices one
at a time.  For each $v\in A_i$, form the child
in~\eqref{eq:candidate-set-update} and call the decision procedure.  Some child
must be feasible; fix the first such $v$ and continue.  After at most $k$
rounds, all labels are assigned.  Across the rounds, at most
$\sum_i|A_i^0|$ vertices are tested, so the number of decision calls is
polynomial in $|\mathcal I|$ and is absorbed
by~\eqref{eq:weighted-runtime}.

For optimization, clear all rational denominators using a common denominator
$D$, whose binary length is polynomial in the input size.  If some original
candidate set is empty, no transversal exists.  Otherwise, let
\[
U=\sum_{i\in[k]}\max_{v\in A_i^0}Dc_i(v).
\]
Every transversal has scaled cost at most $U$.  We first run the decision
procedure with budget $U$; if it returns no, no transversal exists.  Otherwise,
binary search on the integer interval $[0,U]$ finds the least feasible scaled
budget, and self-reduction recovers a transversal of that cost.  Dividing by
$D$ gives the optimum in the original scale.  The bit lengths of $D$ and $U$
and the number of decision calls are polynomial in $|\mathcal I|$.
\end{proof}

\subsection{Completing the proof of the main theorem}

\begin{proof}[Proof of \Cref{thm:main}]
Apply \cref{cor:algorithmic-transfer} to the optimization algorithm of
\cref{thm:weighted}, with $p=d$ and $\mathcal C_d$ the class of
$K_{d,d}$-free graphs.  The resulting \WIT{} instance has $k$ candidate
sets and polynomial encoding length.  The transfer and
\cref{thm:weighted} therefore compute $\OPT(G,S)$ in time
\[
  2^{O(dk\log k)}(n+m)^{O(1)}.
\]
If the optimum is finite, the transfer returns an optimal independent target
and reconstructs a collision-free sequence of exactly that length; if no
transversal exists, it reports that no independent $k$-set is reachable from
$S$.
For the budgeted problem, compare the computed optimum with $b$, incurring
only polynomial dependence on the encoding length of $b$.
\end{proof}

The cases $k\le1$ and $d=1$ are immediate.

\section{Further extensions}
\label{sec:sparsity-extensions}

The proof in \cref{sec:algorithm} has a graph-independent part and a
graph-specific part.  The graph-independent part is the cheap-prefix
recursion: once we know a threshold such that every collection of sufficiently
large candidate sets has compatible representatives, the three cases of the
recursion and their cost argument are unchanged.  The graph-specific part is
only the estimate on conflicts between two candidate sets.  We isolate this
common calculation before deriving four sparsity consequences.

For possibly overlapping sets $U,V\subseteq V(G)$, let $G[U\to V]$ be the
bipartite graph with disjoint parts
\[
 U^{\mathrm L}=\{u^{\mathrm L}:u\in U\},
 \qquad
 V^{\mathrm R}=\{v^{\mathrm R}:v\in V\},
\]
where $u^{\mathrm L}v^{\mathrm R}$ is an edge exactly when $uv\in E(G)$.
The arrow distinguishes the left and right copies; it does not orient $G$.
If $|U|=|V|=\ell$, then $|E(G[U\to V])|$ is the number of ordered choices
$(u,v)\in U\times V$ that conflict by adjacency.  Exactly $|U\cap V|$ further
ordered choices conflict by equality.  Thus
\begin{equation}
 \chi_G(U,V)=|E(G[U\to V])|+|U\cap V|
 \label{eq:universal-conflict-count}
\end{equation}
is the exact number of adjacency-or-equality conflicts among the $\ell^2$
ordered choices.

Suppose that $r$ candidate sets, each of size $\ell$, are replaced by disjoint
copied blocks as in the proof of \cref{lem:sparse-transversal}.  If
\(\chi_G(U,V)\le \rho(\ell)\ell^2\) for every pair of blocks, then at most
\((r-1)\rho(\ell)\ell^2\) conflict edges cross any one block.  Therefore the
Wanless--Wood criterion applies whenever
\begin{equation}
 4(r-1)\rho(\ell)<1.
 \label{eq:universal-ww-certificate}
\end{equation}
The following lemma packages the universal analysis.

\begin{lemma}[Sparse-prefix transfer principle]
\label{lem:sparse-prefix-transfer}
Let $k\ge2$, let $G$ be a graph, and let $\mathcal I$ be a \WIT{} instance
with $k$ candidate sets on $G$.  Suppose that a function
$\gamma:\mathbb N\to\mathbb R_{\ge0}$ satisfies
\[
 |E(G[U\to V])|\le \gamma(\ell)\ell^2
\]
for every $\ell\in\mathbb N$ and all possibly overlapping
$U,V\subseteq V(G)$ with $|U|=|V|=\ell$.  Suppose the algorithm is also given,
for every $r\in\{2,\ldots,k\}$, a positive integer $L(r)$ satisfying
\begin{equation}
 4(r-1)\left(\gamma(L(r))+\frac1{L(r)}\right)<1,
 \label{eq:generic-prefix-certificate}
\end{equation}
and put $\widehat L(k)=\max_{2\le r\le k}L(r)$.  Then the decision version can
be solved deterministically, with a feasible tuple returned for every
yes-instance, in time
\[
 O\!\left((k\widehat L(k))^k|\mathcal I|^{O(1)}\right).
\]
Within the same bound, the optimization version returns a globally
minimum-cost transversal or reports that none exists.
\end{lemma}

\begin{proof}
For two $L(r)$-element candidate sets, adjacency contributes at most
$\gamma(L(r))L(r)^2$ conflicts and equality contributes at most $L(r)$.
Hence their total conflict proportion is at most
$\gamma(L(r))+1/L(r)$.  By~\eqref{eq:universal-ww-certificate}
and~\eqref{eq:generic-prefix-certificate}, every $r$ such sets have distinct,
pairwise nonadjacent representatives.

Run the decision procedure of \cref{sec:algorithm} with $L(r)$ in place of
$M(d,r)$ when $r$ labels remain.  A small candidate set is still enumerated
in full.  If all prefixes are full and collectively affordable, the preceding
certificate supplies a feasible completion.  If they are collectively too
expensive, the sorted-cost argument in \cref{lem:cheap-prefix}(b) still shows
that every feasible completion meets some prefix.  Thus the correctness proof
in \cref{lem:recursion-correct} applies verbatim.

The recursion has depth at most $k$ and branching factor at most
$k\widehat L(k)$.  Each state takes polynomial time, giving the displayed
bound.  The self-reduction and rational-cost scaling in the proof of
\cref{thm:weighted} give a witness and an optimum within the same
parameterized bound.
\end{proof}

\subsection{Edge count: an edge-count-sensitive XP algorithm}
\label{sec:edge-number}

The \emph{edge count} of $G$ is $m=|E(G)|$.  This is a global measure: it
does not prevent all $m$ edges from being concentrated in a comparatively
small vertex set.  Nevertheless, it improves over enumerating all $n^k$
destination tuples when $m\ll n^2$.

For $m\in\mathbb N_0$ and $r\ge2$, define
\begin{equation}
 L_{\mathrm e}(m,r)
 =8(r-1)+\left\lceil4\sqrt{m(r-1)}\right\rceil+1.
 \label{eq:edge-threshold}
\end{equation}

\begin{theorem}[Edge-count-sensitive algorithm]
\label{thm:edge-number}
Let $k\ge2$, and let $\mathcal I$ be a \WIT{} instance with $k$ candidate
sets whose underlying graph $G$ has $m$ edges.  A minimum-cost transversal can be found
deterministically, or nonexistence reported, in time
\[
  2^{O(k\log k)}(m+1)^{k/2}|\mathcal I|^{O(1)}.
\]
For \ISD{}, given an $n$-vertex, $m$-edge graph $G$ and a set
$S\in\binom{V(G)}{k}$, one can compute $\OPT(G,S)$ in time
\[
  2^{O(k\log k)}(m+1)^{k/2}(n+m)^{O(1)}.
\]
If the optimum is finite, the algorithm also returns an optimal independent
$k$-set and a shortest collision-free slide sequence; otherwise, it reports
that no independent $k$-set is reachable from $S$.
In particular, when $m=O(n)$ this is
$2^{O(k\log k)}n^{k/2+O(1)}$ time.
\end{theorem}

\begin{proof}[Proof of \Cref{thm:edge-number}]
Fix $U,V\subseteq V(G)$ with $|U|=|V|=\ell$.  An undirected edge $xy$ of
$G$ contributes at most the two ordered adjacency pairs $(x,y)$ and $(y,x)$
to $U\times V$.  Therefore
\[
 |E(G[U\to V])|\le2m,
 \qquad
 \gamma_{\mathrm e}(\ell)=\frac{2m}{\ell^2}.
\]
For each $r\in\{2,\ldots,k\}$, put $R=r-1$ and
$\ell=L_{\mathrm e}(m,r)$.  Then
$\ell>8R$ and $\ell>4\sqrt{mR}$, so
\[
 4R\left(\frac{2m}{\ell^2}+\frac1\ell\right)
 <\frac12+\frac12=1.
\]
Thus \cref{lem:sparse-prefix-transfer} applies.  The threshold is
nondecreasing in $r$, and
\[
 L_{\mathrm e}(m,k)=O\bigl(k+\sqrt{km}\bigr).
\]
Consequently,
\begin{align*}
 \bigl(kL_{\mathrm e}(m,k)\bigr)^k
 &\le 2^{O(k\log k)}(m+1)^{k/2},
\end{align*}
which proves the \WIT{} bound.  The \ISD{} statement follows from the
algorithmic transfer in \cref{cor:algorithmic-transfer}.  If $m=O(n)$, the
last display is $2^{O(k\log k)}n^{k/2}$, as claimed.
\end{proof}

\subsection{A direct proof for bounded-degeneracy graphs}
\label{sec:degeneracy}

A graph is \emph{$a$-degenerate} if every nonempty subgraph has a vertex of
degree at most $a$.  Equivalently, its vertices can be ordered so that every
vertex has at most $a$ neighbors later in the order.  In particular, every
$h$-vertex subgraph has at most $ah$ edges.

For $a\in\mathbb N_0$ and $r\ge2$, define
\begin{equation}
 L_{\mathrm{deg}}(a,r)=4(r-1)(4a+1)+1.
 \label{eq:degeneracy-threshold}
\end{equation}

\begin{theorem}[Direct bounded-degeneracy algorithm]
\label{thm:degenerate}
Let $a\in\mathbb N_0$ and $k\ge2$.  Given a \WIT{} instance $\mathcal I$
with $k$ candidate sets whose underlying graph $G$ is $a$-degenerate, one can
find a minimum-cost transversal, or report that none exists, in time
\[
  2^{O(k\log k)}(a+1)^k|\mathcal I|^{O(1)}.
\]
For \ISD{}, given an $n$-vertex, $m$-edge, $a$-degenerate graph $G$ and a set
$S\in\binom{V(G)}{k}$, one can compute $\OPT(G,S)$ in time
\[
  2^{O(k\log k)}(a+1)^k(n+m)^{O(1)}.
\]
If the optimum is finite, the algorithm also returns an optimal independent
$k$-set and a shortest collision-free slide sequence; otherwise, it reports
that no independent $k$-set is reachable from $S$.
Consequently, both problems are uniformly FPT in $k+a$.
\end{theorem}

\begin{proof}[Proof of \Cref{thm:degenerate}]
Let $U,V\subseteq V(G)$ have size $\ell$ and put $W=U\cup V$.  Degeneracy
gives $|E(G[W])|\le a|W|\le2a\ell$.  Each edge of $G[W]$ contributes at most
two ordered pairs to $E(G[U\to V])$, and hence
\[
 |E(G[U\to V])|\le4a\ell,
 \qquad
 \gamma_{\mathrm{deg}}(\ell)=\frac{4a}{\ell}.
\]
For each $r\in\{2,\ldots,k\}$, set
$\ell=L_{\mathrm{deg}}(a,r)$.  Then
\[
 4(r-1)\left(\gamma_{\mathrm{deg}}(\ell)+\frac1\ell\right)
 =\frac{4(r-1)(4a+1)}{\ell}<1.
\]
Apply \cref{lem:sparse-prefix-transfer}.  Since
$L_{\mathrm{deg}}(a,k)=O((a+1)k)$, its running time becomes
\[
 O\!\left((kL_{\mathrm{deg}}(a,k))^k|\mathcal I|^{O(1)}\right)
 =2^{O(k\log k)}(a+1)^k|\mathcal I|^{O(1)}.
\]
The \ISD{} clause again follows from \cref{cor:algorithmic-transfer}.
A graph's degeneracy and a corresponding ordering can be computed in
$O(n+m)$ time, so the value $a$ need not be supplied separately.
\end{proof}

Fellows et al.~\cite{FellowsEtAl2026} previously established FPT in $k$ on
every fixed bounded-degeneracy class using independence-covering families.
The argument above is a short and structurally different proof: one local
edge count feeds directly into the cheap-prefix recursion, handles
coordinate-dependent costs and overlapping candidate sets, and makes the
dependence on $a$ explicit.  It is also quantitatively sharper than applying
our biclique-free theorem as a black box.  For $a\ge1$, an $a$-degenerate
graph is $K_{a+1,a+1}$-free, but substituting $d=a+1$ in \cref{thm:main}
gives the factor $2^{O((a+1)k\log k)}$, whereas the direct analysis gives
$2^{O(k\log k+k\log(a+1))}$.

\subsection{Bounded \texorpdfstring{$s$}{s}-codegree}
\label{sec:codegree}

For an integer $s\ge1$, define the \emph{$s$-codegree} of a graph $G$ by
\begin{equation}
 \lambda_s(G)=
 \max_{X\in\binom{V(G)}{s}}\left|\bigcap_{x\in X}N_G(x)\right|,
 \label{eq:s-codegree}
\end{equation}
with $\lambda_s(G)=0$ when $|V(G)|<s$.  Thus $\lambda_1(G)$ is the maximum
degree $\Delta(G)$, and $\lambda_2(G)$ is the usual maximum codegree, the
largest number of common neighbors of a pair of vertices.

Fix $q\in\mathbb N_0$ and assume that $\lambda_s(G)\le q$.  The asymmetric
K\H{o}v\'ari--S\'os--Tur\'an bound~\cite{KovariSosTuran1954} says that a
bipartite graph with $\ell$ vertices in each part and with no
$K_{s,q+1}$ having its $s$-vertex side in the left part has at most
\begin{equation}
 q^{1/s}\ell^{2-1/s}+(s-1)\ell
 \label{eq:asymmetric-kst}
\end{equation}
edges.  For any $\ell$-element sets $U,V\subseteq V(G)$, the copied graph
$G[U\to V]$ satisfies this hypothesis.  Indeed, the $s$ left vertices of a
copied $K_{s,q+1}$ would project to $s$ distinct vertices of $G$ with at
least $q+1$ common neighbors.  No projected vertex can occur on both sides,
because the corresponding copied cross-edge would require a loop in $G$.

For $r\ge2$, define
\begin{equation}
 L_{s,q}(r)=q\bigl(8(r-1)\bigr)^s+8s(r-1)+1.
 \label{eq:codegree-threshold}
\end{equation}

\begin{theorem}[Bounded $s$-codegree]
\label{thm:codegree}
Let $k\ge2$, $s\ge1$, and $q\ge0$ be integers.  Given $s,q$ and a \WIT{}
instance $\mathcal I$ with $k$ candidate sets whose underlying graph $G$ is
promised to satisfy $\lambda_s(G)\le q$, one can find a minimum-cost
transversal, or report that none exists, in time
\[
  2^{O(sk\log k+k\log(q+1))}|\mathcal I|^{O(1)}.
\]
For \ISD{}, given $s,q$, an $n$-vertex, $m$-edge graph $G$ promised to
satisfy $\lambda_s(G)\le q$, and a set $S\in\binom{V(G)}{k}$, one can compute
$\OPT(G,S)$ in time
\[
  2^{O(sk\log k+k\log(q+1))}(n+m)^{O(1)}.
\]
If the optimum is finite, the algorithm also returns an optimal independent
$k$-set and a shortest collision-free slide sequence; otherwise, it reports
that no independent $k$-set is reachable from $S$.
Thus the problems are FPT in $k+s+q$ and, for fixed $s$, in $k+q$.
\end{theorem}

The promise $\lambda_s(G)\le q$ need not be verified.

\begin{proof}[Proof of \Cref{thm:codegree}]
Dividing~\eqref{eq:asymmetric-kst} by $\ell^2$ gives the adjacency-density
bound
\[
 \gamma_{s,q}(\ell)
 =q^{1/s}\ell^{-1/s}+\frac{s-1}{\ell}.
\]
After equality conflicts are included, the expression
in~\eqref{eq:generic-prefix-certificate} is
\[
 q^{1/s}\ell^{-1/s}+\frac{s}{\ell}.
\]
For each $r\in\{2,\ldots,k\}$, put $R=r-1$ and
$\ell=L_{s,q}(r)$.  If $q>0$, then
$\ell>q(8R)^s$ and hence
$q^{1/s}\ell^{-1/s}<1/(8R)$; for $q=0$, this term is zero.  Also
$\ell>8sR$, so $s/\ell<1/(8R)$.  It follows that
\[
 4R\left(q^{1/s}\ell^{-1/s}+\frac{s}{\ell}\right)<1.
\]
The sparse-prefix transfer principle applies.  Since the threshold is
nondecreasing in $r$ and
\[
 L_{s,q}(k)\le q(8k)^s+8sk+1,
\]
we obtain
\[
 \bigl(kL_{s,q}(k)\bigr)^k
 =2^{O(sk\log k+k\log(q+1))}.
\]
This proves the static bound, and \cref{cor:algorithmic-transfer} gives the
\ISD{} bound.
\end{proof}

The case $s=1$ deserves emphasis:~\eqref{eq:s-codegree} becomes
$\lambda_1(G)=\Delta(G)$, and~\eqref{eq:codegree-threshold} is
$O((q+1)r)$.  Thus bounded maximum degree is exactly the first member of the
$s$-codegree hierarchy, not a separate case.  For $s=2$ and $q\ge1$, the
threshold is $O(qr^2+r)$, which can be substantially smaller than treating the
graph only through the balanced biclique $K_{q+1,q+1}$ that it excludes.

\subsection{\texorpdfstring{$K_{s,t}$}{K(s,t)}-free graphs}
\label{sec:unbalanced-biclique}

For integers $1\le s\le t$, the complete bipartite graph $K_{s,t}$ has one
part of size $s$, one part of size $t$, and every possible edge between the
two parts.  A graph is \emph{$K_{s,t}$-free} if it contains no such graph as
a subgraph; the forbidden copy need not be induced.  By the definition of a
common neighborhood,
\begin{equation}
 G\text{ is }K_{s,t}\text{-free}
 \quad\Longleftrightarrow\quad
 \lambda_s(G)\le t-1.
 \label{eq:biclique-codegree-equivalence}
\end{equation}

\begin{corollary}[Unbalanced biclique exclusion]
\label{thm:unbalanced-biclique}
Let $k\ge2$ and let $1\le s\le t$ be integers.  Given a \WIT{} instance
$\mathcal I$ with $k$ candidate sets whose underlying graph $G$ is promised to be
$K_{s,t}$-free, one can find a minimum-cost transversal, or report that none
exists, in time
\[
  2^{O(sk\log k+k\log t)}|\mathcal I|^{O(1)}.
\]
For \ISD{}, given an $n$-vertex, $m$-edge graph $G$ promised to be
$K_{s,t}$-free and a set $S\in\binom{V(G)}{k}$, one can compute $\OPT(G,S)$
in time
\[
  2^{O(sk\log k+k\log t)}(n+m)^{O(1)}.
\]
If the optimum is finite, the algorithm also returns an optimal independent
$k$-set and a shortest collision-free slide sequence; otherwise, it reports
that no independent $k$-set is reachable from $S$.
In particular, both problems are FPT in $k+s+t$ and in $k$ for every fixed
pair $(s,t)$.
\end{corollary}

\begin{proof}[Proof of \Cref{thm:unbalanced-biclique}]
Apply \cref{thm:codegree} with $q=t-1$.  Equivalently, use the explicit
threshold, for every $r\in\{2,\ldots,k\}$,
\[
 L_{s,t}^{\mathrm{bic}}(r)
 =(t-1)\bigl(8(r-1)\bigr)^s+8s(r-1)+1.
\]
Its logarithm is $O(s\log k+\log t)$ for $r\le k$, so the search-tree factor
is $2^{O(sk\log k+k\log t)}$.  The equivalence
in~\eqref{eq:biclique-codegree-equivalence} proves that the promise needed by
\cref{thm:codegree} holds.
\end{proof}

For $t\ge2$, the direct asymmetric analysis strengthens the bound obtained by
treating a $K_{s,t}$-free graph as merely $K_{t,t}$-free.  The latter
observation already gives FPT in $k$ for fixed $s,t$, but it places $t$ in the
extremal exponent;
\cref{thm:unbalanced-biclique} instead puts the smaller side $s$ in that
exponent, while the larger side enters only linearly in the prefix threshold.

\subsection{Weighted movement on a separate graph}

The proof uses biclique exclusion only for terminal conflicts.  The graph in
which tokens move may be different.

\begin{corollary}[Two-graph, weighted movement]
\label{cor:two-graph}
Let $d,k\ge2$ be integers.  Let $G_{\mathrm f}$ be an undirected graph
promised to be $K_{d,d}$-free, and let $G_{\mathrm m}$ be an arbitrary
directed graph on a common $n$-vertex set $V$.
Let $w:E(G_{\mathrm m})\to\mathbb Q_{\ge0}$ be a binary-encoded arc-cost
function shared by all tokens, and let
$S=\{s_1,\ldots,s_k\}\subseteq V$.  A move sends one token along an arc of
$G_{\mathrm m}$ to an unoccupied vertex, and only the terminal configuration
must be independent in $G_{\mathrm f}$.  The total movement cost of a sequence
is the sum, with multiplicity, of the costs of its traversed arcs.  Given
$\mathcal J=(d,G_{\mathrm f},G_{\mathrm m},w,S)$, one can deterministically find
a $G_{\mathrm f}$-independent terminal $k$-set and a collision-free sequence
that lexicographically minimizes
\[
 (\text{total movement cost},\ \text{number of slides})
\]
over all such sequences, in time
$2^{O(dk\log k)}|\mathcal J|^{O(1)}$.  If no such terminal configuration is
reachable, the algorithm reports nonexistence.  Here $|\mathcal J|$ is the
total binary encoding length.
\end{corollary}

\begin{proof}
Multiply all arc costs by the product of their denominators.  This produces
nonnegative integer costs with polynomial encoding length and does not change
the lexicographic optimum; henceforth $w$ denotes these scaled integer costs.
Give an arc $e$ the two-component length
\(\ell(e)=(w(e),1)\), where the first component records cost and the second
records the number of slides.  Vectors are added componentwise and compared
lexicographically; $+\infty$ is larger than every finite vector and is
absorbing under addition.  For $u,v\in V$, let
\(\delta(u,v)\in\mathbb N_0^2\cup\{+\infty\}\) be the lexicographically
minimum length of a directed $u$--$v$ path, or $+\infty$ if no such path
exists.  For two $k$-sets $Q,T$, define
\[
 \mu^{\mathrm{lex}}(Q,T)=
 \min_{\pi:Q\to T\text{ bijective}}
 \sum_{u\in Q}\delta(u,\pi(u)).
\]

We first show that this assignment value equals the lexicographically minimum
cost--slide pair of a collision-free sequence.  Every slide sequence induces
a source--target bijection and therefore has value at least
$\mu^{\mathrm{lex}}(Q,T)$.  Thus no sequence
exists when $\mu^{\mathrm{lex}}(Q,T)=+\infty$.  Suppose henceforth that
$\mu^{\mathrm{lex}}(Q,T)<+\infty$.  If $Q=T$, the empty sequence attains this
value.  Otherwise, fix an optimal bijection $\pi$, choose
$t\in T\setminus Q$, and put $s=\pi^{-1}(t)$.
Follow a lexicographically shortest directed $s$--$t$ path until its first
unoccupied vertex $v$, and let $u$ be the occupied predecessor of $v$.

If $u\ne s$, put $t_u=\pi(u)$ and exchange the destinations of the tokens at
$s$ and $u$.  The directed triangle inequality and the shortest-path equality
$\delta(s,t)=\delta(s,u)+\delta(u,t)$ give
\begin{align*}
 \delta(u,t)+\delta(s,t_u)
 &\le \delta(u,t)+\delta(s,u)+\delta(u,t_u)\\
 &=\delta(s,t)+\delta(u,t_u).
\end{align*}
Thus the exchange cannot increase the assignment value.  Since the original
assignment was optimal, the exchanged one is also optimal and assigns the
token at $u$ to $t$.  If $u=s$, the same conclusion holds without an
exchange.

Now slide the token from $u$ to $v$, obtaining
$Q'=(Q\setminus\{u\})\cup\{v\}$.  Because the $u$--$t$ path begins with the
arc $(u,v)$ and its suffix is shortest, replacing source $u$ by $v$ in the
current assignment gives
\[
 \ell(u,v)+\mu^{\mathrm{lex}}(Q',T)
 \le \mu^{\mathrm{lex}}(Q,T).
\]
Conversely, in any assignment from $Q'$ to $T$, replace source $v$ by $u$
and keep its destination.  The directed triangle inequality gives
\[
 \mu^{\mathrm{lex}}(Q,T)
 \le \ell(u,v)+\mu^{\mathrm{lex}}(Q',T).
\]
Hence
$\mu^{\mathrm{lex}}(Q,T)=\ell(u,v)+\mu^{\mathrm{lex}}(Q',T)$.
Lexicographic order on $\mathbb N_0^2$ is
well-founded, and the second component decreases even when $w(u,v)=0$.
Repeating the argument therefore constructs a collision-free sequence whose
cost--slide pair is exactly $\mu^{\mathrm{lex}}(Q,T)$.

Every lexicographically shortest path is simple: deleting a directed cycle
does not increase its cost and strictly decreases its number of arcs.  Thus
each assigned path uses at most $n-1$ arcs, and the total secondary component
is at most $k(n-1)$.  Set $C=k(n-1)+1$.  For every arc, the scalar length
\[
 Cw(e)+1
\]
preserves the lexicographic order exactly.  Indeed, for any collection of $k$
simple directed paths, the total number $b$ of arcs lies in $[0,C-1]$.  If
$(a,b)$ and $(a',b')$ are the cost--arc vectors of two such collections, then
$Ca+b<Ca'+b'$ if and only if $(a,b)$ is lexicographically smaller than
$(a',b')$: when $a<a'$, integrality gives
$Ca+b\le Ca+C-1<C(a+1)\le Ca'+b'$, and when $a=a'$, the scalar comparison is
exactly the comparison of $b$ and $b'$.  A shortest-path computation from each
source therefore supplies the finite candidate set $A_i$ of
vertices reachable from $s_i$ and, for every $v\in A_i$, the coordinate cost
\[
 c_i(v)=\min_{P:s_i\leadsto v}\sum_{e\in E(P)}\bigl(Cw(e)+1\bigr),
\]
where the minimum ranges over directed $s_i$--$v$ paths in $G_{\mathrm m}$.

Apply \cref{thm:weighted} to these candidate sets in the feasibility graph
$G_{\mathrm f}$.  If no transversal exists, then no reachable independent
terminal configuration exists.  Otherwise, the selected destinations are
distinct and independent in $G_{\mathrm f}$.  For the resulting target set
$T$, a lexicographically cheaper source--destination bijection would itself
define a cheaper feasible \WIT{} tuple, contradicting optimality.  Hence the
selected pairing attains $\mu^{\mathrm{lex}}(S,T)$, and the directed
first-unoccupied-vertex construction above recovers an optimal collision-free
sequence.  The scaling, shortest paths, and reconstruction all have
polynomial bit complexity and polynomial running time outside the call to
\cref{thm:weighted}, which proves the stated bound.
\end{proof}

\end{document}